\documentclass[conference]{IEEEtran}
\IEEEoverridecommandlockouts
\usepackage{cite}
\usepackage{amsmath,amssymb,amsfonts}
\usepackage{algorithmic}
\usepackage{graphicx}
\usepackage{textcomp}
\usepackage{xcolor}
\usepackage{booktabs}
\usepackage{tabularx}
\usepackage{multirow}
\usepackage{pgfplots}
\usepackage{pgfplotstable}
\pgfplotsset{compat=1.18}
\usepgfplotslibrary{statistics}
\usepackage{amsthm}
\usepackage{hyperref}
\newtheorem{lemma}{Lemma}

\def\BibTeX{{\rm B\kern-.05em{\sc i\kern-.025em b}\kern-.08em
    T\kern-.1667em\lower.7ex\hbox{E}\kern-.125emX}}

\usepackage{todonotes}

\begin{document}

\title{Why Is SHAP Not a Reliable Standalone Explanation Framework for Malware Detection?}

\author{
\IEEEauthorblockN{Seyedreza Mohseni\IEEEauthorrefmark{1},
Edward Raff\IEEEauthorrefmark{2}\IEEEauthorrefmark{1},
Manas Gaur\IEEEauthorrefmark{1}}
\IEEEauthorblockA{\IEEEauthorrefmark{1}\textit{University of Maryland Baltimore County}, Baltimore, USA}
\IEEEauthorblockA{mohseni1@umbc.edu, manas@umbc.edu}
\IEEEauthorblockA{\IEEEauthorrefmark{2}\textit{CrowdStrike}, New York, USA, edward.raff@crowdstrike.com}
}

\maketitle

\begin{abstract}
Machine learning is widely used for malware detection, but its decisions must be explained. An analyst needs to know whether a model has learned genuine malicious behavior or only dataset-specific patterns \cite{gaur2021semantics}. SHapley Additive exPlanations (SHAP) is the standard tool for this, backed by formal properties such as local accuracy, missingness, and consistency. We argue that these guarantees are insufficient for reliable malware interpretation. We claim SHAP explains a chosen feature-coalition game, not malware behavior in the data. That game is fixed only after the analyst selects the feature players, the missing feature rule, the background distribution, and the simplified input mapping. In static Portable Executable feature spaces, groups such as byte histograms, byte-entropy, strings, headers, sections, imports, and data-directories are not independent signals but are jointly shaped by file structure, packing, compiler behavior, and family conventions. We prove that this dependence makes conditional SHAP dilute a model's feature credit by a factor of $1/m$ across $m-1$ redundant features, attributes importance to features the model never uses, and even reverses the sign of an unused feature's attribution when the data distribution changes; interventional SHAP, meanwhile, queries off-manifold coalitions that no real executable would exhibit. Experiments on EMBER-2018, EMBER-2024, and BODMAS with fixed LightGBM and XGBoost detectors confirm these effects. We therefore position SHAP as a limited diagnostic that requires an explicitly stated data distribution and domain validation, not a standalone account of malware behavior.
\end{abstract}


\section{Introduction}

Machine learning is now a core component of modern malware detection, processing large volumes of files and uncovering discriminative patterns that are difficult for human analysts to identify manually~\cite{saxe2015deep,arp2014drebin,raff2018malconv,raff_cybersecurity_2026}. In static analysis of Portable Executable (PE) files, gradient-boosted trees and deep networks routinely achieve high detection rates and are increasingly deployed in operational security pipelines~\cite{Raff2020a}. In these settings, raw accuracy is not sufficient. An analyst who must triage an alert, justify a quarantine, or audit a model before deployment needs to know \emph{why} a file was flagged, whether the model relies on genuine malicious behavior or on incidental dataset artifacts, and whether the same explanation would hold for similar files~\cite{rudin2019we,warnecke2020evaluating,nandan_towards_2026,stortz_large_2026}. This demand for accountability has made post-hoc feature attribution a standard part of the malware-ML toolchain~\cite{ribeiro2016why,sundararajan2017axiomatic,Raff2020autoyara,Zak2017}.

Among these methods, SHapley Additive exPlanations (SHAP)~\cite{lundberg2017unified} is the most widely adopted. Rooted in cooperative game theory~\cite{shapley1953value}, SHAP assigns each feature an additive importance value and unifies a family of attribution methods under formal properties such as local accuracy, missingness, and consistency. Efficient estimators such as TreeSHAP~\cite{lundberg2018consistent} and KernelSHAP~\cite{covert2021improving} have lowered the cost of applying SHAP to high-dimensional malware feature spaces, and SHAP attributions now frequently serve as de facto evidence that a classifier has learned the right thing.

This paper argues that such trust is misplaced when SHAP is used as a \emph{standalone} interpreter of static PE malware models. Our concern is not that SHAP computes incorrect numbers. For a fixed value function and simplified input mapping, its attributions are well-defined, and its axioms hold. Our concern is not the numerical validity of SHAP attributions for a fixed value function and simplified input mapping. Rather, it is the semantic object those attributions characterize: SHAP explains the selected feature-coalition game, which need not correspond to the malware behavior represented by the data~\cite{gaur2024building}. That game is fixed only after the analyst has selected the feature players, the missing feature rule, the background distribution, and the simplified input mapping; in the malware setting, these choices are neither neutral nor obvious, and different choices yield different attributions for the same model and file~\cite{sundararajan2020many,merrick2019explanation}.

This gap between the features and malware matters more for malware than for many domains because of \emph{feature dependence}. In static PE feature representations such as EMBER-2018~\cite{anderson2018ember}, EMBER-2024 ~\cite{joyce2025ember2024} and BODMAS~\cite{yang2021bodmas}, the standard feature groups such as \texttt{byte-histograms}, \texttt{byte-entropy}, \texttt{strings}, \texttt{headers}, \texttt{sections}, \texttt{imports}, and \texttt{data-directories} are not independent signals. They come from the same executable and are jointly shaped by file structure, packing, compiler behavior, family conventions, and time-dependent collection effects known to bias malware benchmarks~\cite{pendlebury2019tesseract} (e.g., obfuscations transform all features together~\cite {mohseni_can_2025,mohseni_analyzing_2026}). SHAP's coalition mechanism, however, treats features as separable players that can be added to or removed from a coalition. Under such dependence, the two ways of realizing this mechanism diverge, and each becomes problematic.

\emph{Interventional} SHAP replaces missing features with background draws, producing combinations that no real PE file would exhibit and evaluating the model in regions it never learned~\cite{janzing2020feature,hooker2021unrestricted}; we show empirically that the measured group-level dependence in these datasets makes such off-manifold coalitions the rule rather than the exception. \emph{Conditional} SHAP respects the data distribution but can instead credit proxy features that merely predict others, absorb dataset artifacts such as family signatures, and require estimating high-dimensional conditional densities that are themselves error-prone~\cite {aas2021explaining,frye2020asymmetric}; we make these failures precise as lemmas in Section~\ref{sec:theory}. The analyst is thus pushed toward either an explanation of artificial interventions or one of indirect statistical association, neither of which is guaranteed to describe the malicious behavior under study.

These tensions are not new in general machine learning. Kumar et al.~\cite{kumar2020problems} show that Shapley-based importance behaves counterintuitively under dependence and redundancy in features. Janzing et al.~\cite{janzing2020feature} frame feature relevance as a causal problem motivating the interventional formulation; Chen et al.~\cite{chen2020true} crystallize the choice as true to the model or true to the data; and Sundararajan and Najmi~\cite{sundararajan2020many} show that incompatible value functions yield many Shapley values for the same prediction. Reliability concerns extend even to adversarial manipulation of SHAP outputs~\cite{slack2020fooling}, and global Shapley-style importance inherits the same dependence sensitivities~\cite{covert2020understanding}. In security specifically, Warnecke et al.~\cite{warnecke2020evaluating} found that explanation methods struggle with structured security data and produce unstable outputs. What has been missing is a \emph{focused, mathematically explicit} account of how these abstract failure modes manifest on the static PE feature groups the community actually uses, and what they imply for an analyst reading a SHAP plot as evidence of malware behavior. This paper provides that account, a set of lemmas for the interventional and conditional cases with the underlying classifiers held fixed, together with a controlled empirical study, so that any instability is attributable to the explanation process rather than to model quality.

\textbf{Contributions}. This article supports its results through a combination of theoretical exposition showing that multiple issues can occur, followed by empirical validation. 
\begin{itemize}
  \item \textbf{A domain-specific articulation of the problem.} We formalize, for static PE malware feature groups, \emph{why} interventional and conditional SHAP diverge under feature dependence, and we identify the resulting failure modes: off-manifold coalition queries for interventional SHAP, and redundancy dilution, proxy crediting, and sign instability for conditional SHAP.
  
  \item \textbf{An interventional off-manifold failure (Lemma~1).} We prove that whenever PE feature groups are statistically dependent ($I_P(X_C; X_{\bar C})>0$), interventional SHAP evaluates the classifier on synthetic group-mixed coalitions whose distribution differs from the real data distribution. 

  \item \textbf{Conditional attribution failures (Lemmas~2--4).} We prove that conditional SHAP (i) dilutes the attribution of a model-used feature by a factor $1/m$ when $m-1$ redundant features (proxies) encoding the same latent malware factor are added (Lemma~2, \emph{redundancy}); (ii) assigns nonzero attribution to a feature that the classifier never uses whenever that feature is statistically informative about a used feature (Lemma~3, \emph{proxy attribution}); and (iii) can reverse the attribution \emph{sign} of such an unused feature across two valid dataset distributions (Lemma~4, \emph{sign instability}).

  \item \textbf{A controlled empirical demonstration.} On EMBER-2018~\cite{anderson2018ember}, EMBER-2024 ~\cite{joyce2025ember2024}, and BODMAS~\cite{yang2021bodmas}, we (i) measure the group-level dependence structure and show that PE feature groups are far from the independent players that SHAP's coalition mechanism assumes, and (ii) for \textbf{EMBER-2018}, with strong LightGBM and XGBoost detectors held fixed, we sweep injected redundant features and observe attribution migrating off the model used features in close agreement with the $1/m$ prediction of Lemma~2, while predictive performance remains unchanged, isolating the instability to the explanation process rather than the detector.
  
\end{itemize}

\section{Background of SHAP}

SHAP explains a model's prediction by treating the input features as players in a game and assigning each feature a value, called a Shapley value, which reflects its contribution to the final prediction. For a model $f$, an input $x$, and a set of features $D$, the SHAP value for feature $i$ is usually written as 

$$
\setlength{\abovedisplayskip}{2pt}
\phi_i(x)=\sum_{S\subseteq D\setminus\{i\}}\frac{|S|!(d-|S|-1)!}{d!}\Big[v_x(S\cup\{i\})-v_x(S)\Big]
\label{eq:main-shap}
$$

where $v_x(S)$ ($S\subseteq D$) is the value of the model when only the features in $S$ are known; three main properties support SHAP. 
Assume $z'\in{\{0,1}\}^d$ is a simplified binary input, $\phi_0$ is the base value, and each $\phi_i$ is the attribution assigned to feature $i$. We can define the explanation model $g(z')$ as $g(z')=\phi_0+\sum_{i=1}^{d}\phi_i z'_i$. The SHAP paper defines additive feature attribution methods as explanation models that are linear functions of binary variables.

Missingness means that a feature that is not present should receive no credit, so if $z'_i=0$, then $\phi_i=0$. Consistency means that if a feature contributes more to a model after the model changes, its assigned value should not decrease. The original SHAP paper shows that, under these conditions, there is a unique additive attribution method for a chosen simplified input mapping $h_x(z')$. However, these are the key limitations for malware analysis: \textit{1- How do we define $v_x(S)$?} \textit{2- How do we map missing features back to real inputs?} \textit{3- How do we handle the conditional relationship between features?}

If the value function uses conditional expectations, interventional sampling, or an independence approximation, then the resulting SHAP values can change. Therefore, the formal SHAP guarantees hold only after the malware analyst has chosen the feature players, the missing-feature rule, the background distribution, and the simplified input mapping. This is especially important for malware datasets, where features are dependent, and removing one feature may not represent a real executable file.

\section{Theoretical Explanation}
\label{sec:theory}


\subsection{Failure of Interventional SHAP}

\begin{lemma}[Invalid Coalition Failure of Interventional SHAP in Malware Feature Spaces]
\label{lem:invalid_coalition_interventional_malware}
Let \(X\sim P_{\mathrm{data}}\) be a real static malware feature vector extracted from PE files. Assume that the feature vector is partitioned into malware-relevant feature groups $\mathcal{G} = \{g_1,\ldots,g_m \}$. For a nonempty proper group coalition \(C\subset\mathcal{G}\), define a synthetic group-mixed coalition sample by $\widetilde{X}^{C} = \left( X_C^{a},X_{\bar C}^{b} \right)$ ($\widetilde{X}^{C}$ representing the entire feature set). Let \(Q_C\) denote the distribution of \(\widetilde{X}^{C}\). Then $Q_C = P_{\mathrm{data}}(X_C)\otimes P_{\mathrm{data}}(X_{\bar C})$ ($\otimes$ denotes the product measure) If the malware feature groups are statistically dependent, $I_{P}(X_C; X_{\bar C})>0 $, then $Q_C\neq P_{\mathrm{data}}$ Therefore, interventional SHAP can evaluate the malware classifier on synthetic group coalitions that are not distributed like real malware or benign PE feature vectors.
\end{lemma}

\begin{proof}
The real malware dataset distribution over the selected and missing feature groups is $P_{\mathrm{data}}(X_C, X_{\bar C})$. However, the group mixed coalition construction draws \(X_C\) from one real file and \(X_{\bar C}\) from another independent real file. Therefore, the synthetic coalition distribution is the product of the two marginal distributions $Q_C = P_{\mathrm{data}}(X_C)P_{\mathrm{data}}(X_{\bar C})$. The mutual information between the two group blocks is

\[
\setlength{\abovedisplayskip}{2pt}
I_{P}(X_C;X_{\bar C}) = D_{\mathrm{KL}} \left(P_{\mathrm{data}}(X_C,X_{\bar C})
\,\middle\|\, P_{\mathrm{data}}(X_C)P_{\mathrm{data}}(X_{\bar C}) \right)
\]

By the non-negativity property of KL divergence, $I_{P}(X_C;X_{\bar C})=0$ if and only if
\[
\setlength{\abovedisplayskip}{2pt}
P_{\mathrm{data}}(X_C,X_{\bar C}) = P_{\mathrm{data}}(X_C)P_{\mathrm{data}}(X_{\bar C})
\]

Thus, if $I_{P}(X_C;X_{\bar C})>0$ then

\[
\setlength{\abovedisplayskip}{2pt}
P_{\mathrm{data}}(X_C,X_{\bar C}) \neq P_{\mathrm{data}}(X_C)P_{\mathrm{data}}(X_{\bar C})
\]

and therefore $Q_C\neq P_{\mathrm{data}}$. This distributional mismatch directly affects interventional SHAP. The group-level interventional value function is $V_x^{int}(C) = \mathbb{E}_{Z_{\bar C}\sim P_{\mathrm{data}}(X_{\bar C})} \left[f(x_C,Z_{\bar C}) \right] $ ~\cite{janzing2020feature}. This value function evaluates the classifier on hybrid samples of the form $(x_C, Z_{\bar C})$ where \(x_C\) comes from the explained file and \(Z_{\bar C}\) is sampled from the marginal distribution of missing groups. However, the real data-consistent conditional value is

\[
\setlength{\abovedisplayskip}{2pt}
V_x^{cond}(C) = \mathbb{E}_{Z_{\bar C}\sim P_{\mathrm{data}}(X_{\bar C}\mid X_C=x_C)} \left[f(x_C,Z_{\bar C}) \right]
\]

When the malware feature groups are dependent,

\[
\setlength{\abovedisplayskip}{2pt}
P_{\mathrm{data}}(X_{\bar C}\mid X_C=x_C) \neq P_{\mathrm{data}}(X_{\bar C})
\]

Consequently,

$$
\setlength{\abovedisplayskip}{2pt}
V_x^{cond}(C)-V_x^{int}(C) = \int f(x_C,z) \left[
p_{\mathrm{data}}(z\mid x_C)-p_{\mathrm{data}}(z) \right]dz
$$

which need not be zero. Therefore, interventional SHAP may compute marginal contributions using artificial coalition samples rather than realistic PE feature vectors. Interventional SHAP fills missing features from their marginal distribution, breaking the dependence that binds PE feature groups within a single executable. It thus queries the malware classifier on synthetic, off-manifold coalitions that no real benign or malicious file would produce. Hence, its attributions reflect model behavior in unlearned regions of feature space rather than the malicious behavior an analyst seeks to interpret.
\end{proof}


\subsection{Failure of Conditional SHAP}

Conditional SHAP is often presented as a safer alternative to interventional SHAP because it respects the data distribution; instead of replacing missing features independently, it defines the coalition value as $V_x^{\mathrm{cond}}(S)=\mathbb{E}[f(X)\mid X_S=x_S]$. In principle, this avoids some artificial feature combinations. However, it does not make conditional SHAP a reliable explanation method for malware. The reason is the strong dependence among PE features discussed above: most feature groups are linked, repeated, or jointly shaped by file structure, so conditional SHAP explains a model over highly correlated variables. Therefore, the conditional distribution
$P(X_{\bar S}\mid X_S=x_S)$ is not a minor technical detail; it becomes the core object that determines the explanation. We identify three resulting failure modes.

In static PE malware representations, a single latent factor, such as packing, encryption, obfuscation, polymorphic mutation, metamorphic code rewriting, or family-specific implementation style, often manifests through multiple observable features simultaneously. For example, a packing-related latent factor $Z$ may be encoded in a byte-entropy feature $R_1$ as well as in related histogram, opcode-frequency, \texttt{section}, \texttt{string}, or API-pattern features $R_2,\ldots, R_m$. Because these features are extracted from the same executable, they carry redundant information about $Z$ and are not independent players. Lemma \autoref{lem:redundancy_failure_conditional_shap_malware} shows that this redundancy causes conditional SHAP to dilute and misattribute credit.

\begin{lemma}[Redundancy Failure of Conditional SHAP]
\label{lem:redundancy_failure_conditional_shap_malware}
Let $R=(R_1,\ldots,R_m)$ with $R_j=h_j(Z)$ for injective $h_j$, so that $P(Z\mid R_j = h_j(z)) = \delta_z$. Let $f(R_1,\ldots,R_m)=g(R_1)$ for measurable $g$, so $\frac{\partial f}{\partial R_j}=0$ for every $j\geq 2$ whenever the derivative is defined. Let Conditional SHAP use the value function $v_x^{cond}(S) = \mathbb{E} \left[ f(R_1,\ldots,R_m)\mid R_S=r_S \right]$.

For a sample \(x\) with \(Z=z\), the Conditional SHAP value assigned to every feature is

\[
\setlength{\abovedisplayskip}{2pt}
\phi_{R_j}^{cond}(x) = \frac{g(h_1(z))-\mathbb{E}[g(h_1(Z))]}{m} \qquad j=1,\ldots,m
\setlength{\belowdisplayskip}{2pt}
\]

Therefore, Conditional SHAP assigns nonzero attribution to features \(R_j\) and \(j\geq 2\), which the classifier does not use, and the attribution of the used feature \(R_1\) is reduced by a factor of \(1/m\). Conditional SHAP is not invariant to redundant feature representation.
\end{lemma}

\begin{proof}
Since each redundant malware feature satisfies $R_j=h_j(Z)$ and each \(h_j\) is injective, observing any nonempty subset of the redundant features reveals the same latent malware factor \(Z=z\). Therefore, for any nonempty coalition \(S\neq\emptyset\),

\[
\setlength{\abovedisplayskip}{2pt}
V_x^{cond}(S) = \mathbb{E} \left[ g(R_1)\mid R_S=r_S \right] = g(h_1(z))
\setlength{\belowdisplayskip}{2pt}
\]

For the empty coalition, $v_x^{cond}(\emptyset) = \mathbb{E}[g(R_1)] = \mathbb{E}[g(h_1(Z))]$. The Conditional SHAP value of feature \(R_j\) is

\[
\begin{gathered}
\phi_{R_j}^{cond}(x) = \sum_{S\subseteq M\setminus\{j\}} \frac{|S|!(m-|S|-1)!}{m!} \\
\left[ v_x^{cond}(S\cup\{j\})-v_x^{cond}(S) \right]
\end{gathered}
\]

where $M=\{1,\ldots,m\}$. If \(S=\emptyset\), then $v_x^{cond}(\{j\})-v_x^{cond}(\emptyset) = g(h_1(z))-\mathbb{E}[g(h_1(Z))]$.
If \(S\neq\emptyset\), then both \(S\) and \(S\cup\{j\}\) reveal the same latent factor \(Z=z\), so $v_x^{cond}(S\cup\{j\})-v_x^{cond}(S) = g(h_1(z))-g(h_1(z)) = 0$.

Thus, the only nonzero marginal contribution occurs when \(S=\emptyset\). The Shapley weight of the empty coalition is $\frac{0!(m-1)!}{m!} = \frac{1}{m}$. Therefore,

\[
\setlength{\abovedisplayskip}{2pt}
\phi_{R_j}^{cond}(x) = \frac{1}{m} \left[ g(h_1(z))-\mathbb{E}[g(h_1(Z))] \right]
\qquad j=1,\ldots,m
\]

This proves that all redundant features receive equal attribution, including features that the classifier does not directly use. If only the model-used feature \(R_1\) were present, its single-feature attribution would be $\phi_{R_1}^{single}(x) = g(h_1(z))-\mathbb{E}[g(h_1(Z))]$.
After adding \(m-1\) redundant malware features, its attribution becomes $\phi_{R_1}^{cond}(x) = \frac{1}{m} \left[ g(h_1(z))-\mathbb{E}[g(h_1(Z))] \right]$.
Hence, $\phi_{R_1}^{cond}(x)/\phi_{R_1}^{single}(x) = 1/m$.
As \(m\) increases, $\phi_{R_1}^{cond}(x)\rightarrow 0$, even though the classifier and the underlying malware signal are unchanged. Therefore, Conditional SHAP attribution depends on how many redundant feature encodings are included in the malware feature table. This proves the redundancy failure.
\end{proof}


Beyond exact redundancy, PE feature spaces often include features that are merely statistically correlated with a model-used feature, rather than deterministic functions of the same latent factor. For instance, a classifier may use only a byte-entropy measurement $B$. In contrast, a related histogram bin, opcode-frequency statistic, API-call pattern, CFG metric, section statistic, or another entropy-derived measurement $R$ is included in the feature table. Even though the classifier is functionally independent of $R$, statistical dependence between $B$ and $R$ in the data allows conditional SHAP to assign nonzero attribution to $R$, as Lemma \autoref{lem:proxy_attribution_failure_conditional_shap_malware} shows.

\begin{lemma}[Proxy Attribution Failure of Conditional SHAP]
\label{lem:proxy_attribution_failure_conditional_shap_malware}
Let $f(B,R)=g(B)$ for integrable $g$, so $\frac{\partial f}{\partial R}=0$ whenever the derivative is defined. Let Conditional SHAP use the value function $v_x^{cond}(S) = \mathbb{E}\left[f(B,R)\mid X_S=x_S\right]$. For a sample \(x=(b,r)\), the Conditional SHAP attribution assigned to \(R\) is

\begin{equation}
\small
\phi_R^{cond}(x) = \frac{1}{2} \left[\mathbb{E}[g(B)\mid R=r] - \mathbb{E}[g(B)] \right]
\label{eq:conditional}
\end{equation}

Therefore, if $\mathbb{E}[g(B)\mid R=r]\neq \mathbb{E}[g(B)]$, then $\phi_R^{cond}(x)\neq 0$ even though the classifier does not use \(R\). Conditional SHAP can assign attribution to a feature solely because it is statistically informative about another feature used by the model.
\end{lemma}

\begin{proof}
For two features, the Conditional SHAP attribution of \(R\) is

\[
\setlength{\abovedisplayskip}{2pt}
\phi_R^{cond}(x) = \frac{1}{2} \left[v(\{R\})-v(\emptyset) \right] + \frac{1}{2} \left[v(\{B,R\})-v(\{B\}) \right]
\]

Since the classifier satisfies $f(B,R)=g(B)$, we have $v(\emptyset) = \mathbb{E}[g(B)]$, $v(\{R\}) = \mathbb{E}[g(B)\mid R=r]$, $v(\{B\}) = \mathbb{E}[g(B)\mid B=b] = g(b)$, and $v(\{B,R\}) = f(b,r) = g(b)$. Therefore, $v(\{B,R\})-v(\{B\}) = g(b)-g(b) = 0$. Substituting these values into the Conditional SHAP formula gives Eq. ~\ref{eq:conditional} Thus, whenever conditioning on \(R=r\) changes the expected value of the model-used feature response \(g(B)\), the proxy feature receives nonzero attribution:

\[
\setlength{\abovedisplayskip}{2pt}
\mathbb{E}[g(B)\mid R=r]\neq \mathbb{E}[g(B)] \quad\Longrightarrow\quad \phi_R^{cond}(x)\neq 0
\]

However, the classifier is functionally independent of \(R\). The nonzero attribution therefore arises from statistical dependence in the malware dataset, not from the classifier's direct use of \(R\). This proves the proxy attribution failure.
\end{proof}

Lemma~\ref{lem:proxy_attribution_failure_conditional_shap_malware} generalizes the $m=2$ case of Lemma~\ref{lem:redundancy_failure_conditional_shap_malware} by relaxing the requirement from exact deterministic redundancy (injective $h_j$) to mere statistical dependence between $R$ and $B$. In Lemma~\ref{lem:redundancy_failure_conditional_shap_malware}, conditioning on $R$ reveals $Z$ exactly and hence determines $B$, so $\mathbb{E}[g(B)\mid R=r] = g(b)$ and the two-player attribution reduce to exactly $\frac{1}{2}[g(b)-\mathbb{E}[g(B)]]$. In Lemma~\ref{lem:proxy_attribution_failure_conditional_shap_malware}, the proxy $R$ needs only to be statistically informative about $B$, producing a potentially weaker but still nonzero attribution. The exact redundancy setting thus represents the worst case within the broader proxy attribution failure.


The proxy attribution failure of Lemma \autoref{lem:proxy_attribution_failure_conditional_shap_malware} is not merely quantitative; it can also reverse sign across dataset distributions. Lemma \autoref{lem:sign_instability_conditional_shap_malware} shows that this causes the conditional SHAP sign of $R$ to flip.

\begin{lemma}[Sign Instability Failure of Conditional SHAP]
\label{lem:sign_instability_conditional_shap_malware}
Let $f(B,R)=g(B)$ for integrable $g$, so $\frac{\partial f}{\partial R}=0$ whenever the derivative is defined. Let Conditional SHAP use the value function $v_{x,P}^{cond}(S) = \mathbb{E}_{P}\left[f(B,R)\mid X_S=x_S\right]$, where \(P\) is the dataset distribution. For a sample \(x=(b,r)\), the Conditional SHAP value of \(R\) is $\phi_{R,P}^{cond}(x) = \frac{1}{2} \left(\mathbb{E}_{P}[g(B)\mid R=r] - \mathbb{E}_{P}[g(B)] \right)$ Therefore,

$$
\setlength{\abovedisplayskip}{2pt}
\operatorname{sign}\left(\phi_{R,P}^{cond}(x)\right) = \operatorname{sign}\left(\mathbb{E}_{P}[g(B)\mid R=r] - \mathbb{E}_{P}[g(B)]\right)
$$

Consequently, for the same classifier \(f(B, R)=g(B)\) and the same feature \(R\) not used by the model, two valid malware dataset distributions \(P^{+}\) and \(P^{-}\) can produce opposite Conditional SHAP signs if $\mathbb{E}_{P^{+}}[g(B)\mid R=r] > \mathbb{E}_{P^{+}}[g(B)]$ and $\mathbb{E}_{P^{-}}[g(B)\mid R=r] < \mathbb{E}_{P^{-}}[g(B)]$. In that case, $\phi_{R,P^{+}}^{cond}(x)>0$ and $\phi_{R,P^{-}}^{cond}(x)<0$.

Thus, Conditional SHAP can assign opposite signs to an unused malware feature solely because the dataset dependence between \(R\) and the model-used feature \(B\) changes.
\end{lemma}

\begin{proof}
For two features, the Conditional SHAP attribution of \(R\) is

\[
\setlength{\abovedisplayskip}{2pt}
\phi_{R,P}^{cond}(x) = \frac{1}{2} \left[ v(\{R\})-v(\emptyset) \right] + \frac{1}{2} \left[ v(\{B,R\})-v(\{B\})\right]
\]

Since the classifier depends only on \(B\), we have
$v(\emptyset) = \mathbb{E}_{P}[g(B)]$,
$v(\{R\}) = \mathbb{E}_{P}[g(B)\mid R=r]$,
$v(\{B\}) = \mathbb{E}_{P}[g(B)\mid B=b] = g(b)$,
and $v(\{B,R\}) = f(b,r) = g(b)$.
Therefore, $v(\{B,R\})-v(\{B\}) = g(b)-g(b) = 0$.
Substituting into the Conditional SHAP formula gives

\[
\setlength{\abovedisplayskip}{2pt}
\phi_{R,P}^{cond}(x) = \frac{1}{2} \left(\mathbb{E}_{P}[g(B)\mid R=r] - \mathbb{E}_{P}[g(B)] \right)
\]

Thus, the sign of the Conditional SHAP value assigned to \(R\) is determined by whether conditioning on \(R=r\) increases or decreases the expected value of the model-used feature response \(g(B)\). If $\mathbb{E}_{P}[g(B)\mid R=r] > \mathbb{E}_{P}[g(B)]$ then $\phi_{R,P}^{cond}(x)>0$. If $\mathbb{E}_{P}[g(B)\mid R=r] < \mathbb{E}_{P}[g(B)]$ then $\phi_{R,P}^{cond}(x)<0$. This proves that the attribution sign of \(R\) can change even though the classifier has no direct dependence on \(R\). The sign is caused by the conditional data law \(P(B\mid R=r)\), not by direct model use of \(R\).
\end{proof}

\section{Practical Experiments}

\subsection {Experimental environments and features}
We implemented all experiments in Google Colab using an NVIDIA A100 GPU. We used three malware datasets, EMBER-2018 \footnote{https://github.com/elastic/ember}, EMBER-2024 \footnote{https://github.com/futurecomputing4ai/ember2024}, and BODMAS \footnote{https://github.com/whyisyoung/BODMAS}, and we limited the feature space to static PE features that are common and meaningful for malware analysis. We excluded non-feature information, such as hashes and other metadata that do not directly describe the executable structure. Code is available in the Malware-SHAP repository \footnote{https://github.com/MohseniMalwareLab/Malware-SHAP} for reproducibility.


\begin{figure}[t]
    \centerline{\includegraphics[width=0.9\columnwidth]{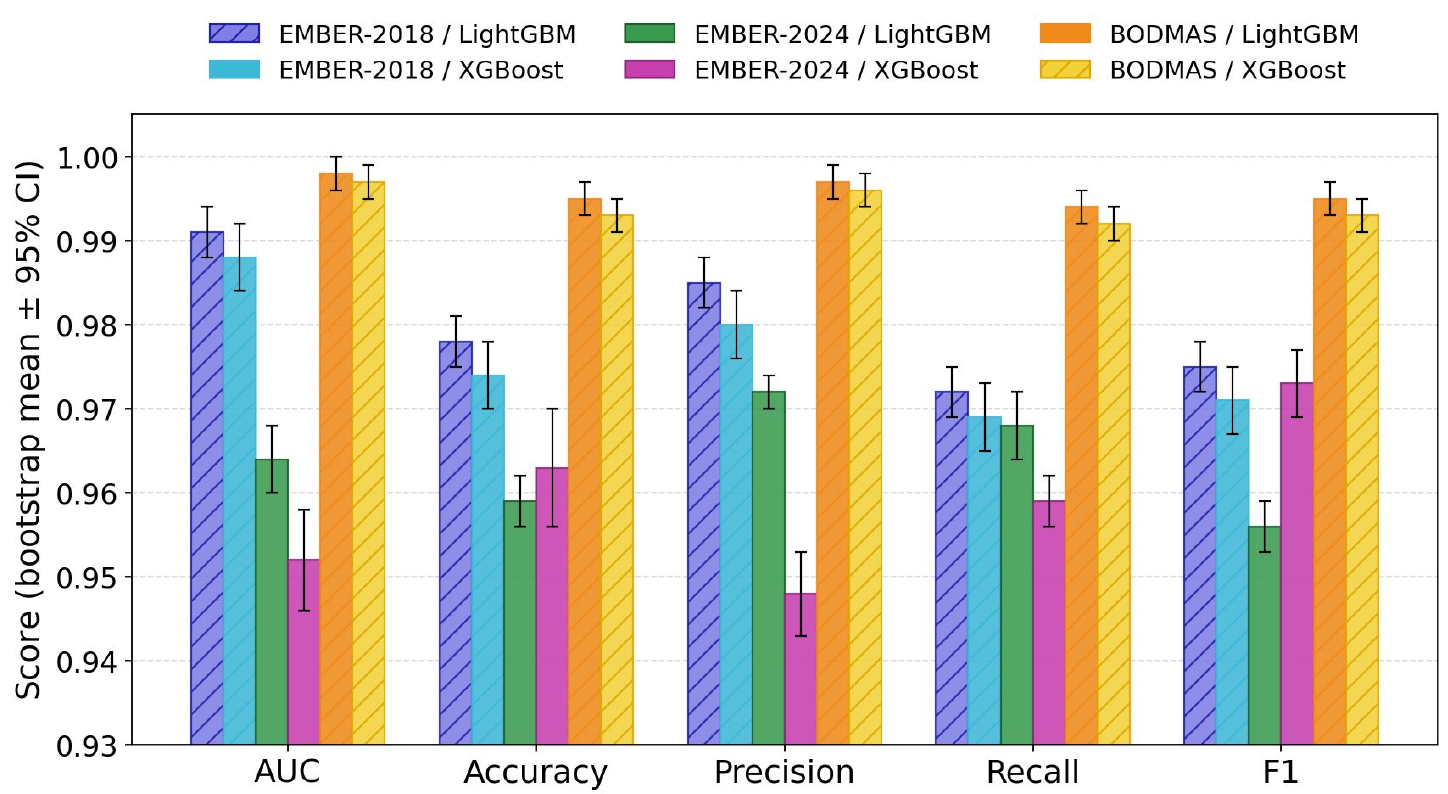}}
    \caption{Detection performance of the two fixed baseline detectors, LightGBM and XGBoost, on EMBER-2018, EMBER-2024, and BODMAS, reported as bootstrap means with 95\% confidence intervals over five seeds. \textbf{Takeaway:} Because the classifiers being explained are uniformly accurate, any subsequent change in SHAP attribution must be attributed to the explanation settings rather than to poor model quality.}
    \label{fig:baseline}
\end{figure}

\subsection{Experiments}
In this section, we briefly present a series of experiments to demonstrate the performance of the SHAP explanation process in malware classification. We used LightGBM with num\_leaves = 64 and XGBoost with max\_leaves = 64, with a learning rate of 0.05 and five different seeds = \{2026, 2027, 2028, 2029, 2030\}. The key takeaway is that these baseline models serve as fixed base detectors, not the main scientific claim. The purpose is to test whether SHAP explanations remain reliable when the classifier is already strong. Since both models perform well, later changes in SHAP attribution can be attributed to explanation settings such as the value function, background distribution, proxy features, or invalid coalitions, rather than to poor classifier quality.

\textbf{Dependencies:} The first experiment measured dependency in the feature space. For each pair of feature groups $G_a$ and $G_b$, we computed the group-level mean absolute correlation. This experiment tested whether malware features behave like independent players. If the feature groups were independent, then $P(X_{\bar S}\mid X_S=x_S)\approx P(X_{\bar S})$. Our motivation was to show that this condition is unrealistic for PE malware data because byte content, entropy, \texttt{strings}, \texttt{headers}, \texttt{imports}, \texttt{sections}, and \texttt{data-directories} are all linked through the same executable file structure. The dependency experiment is a verification of Lemma ~\ref{lem:invalid_coalition_interventional_malware}, ~\ref{lem:proxy_attribution_failure_conditional_shap_malware} and ~\ref{lem:sign_instability_conditional_shap_malware}. 

\begin{figure*}[t]
    \centerline{\includegraphics[width=0.85\textwidth]{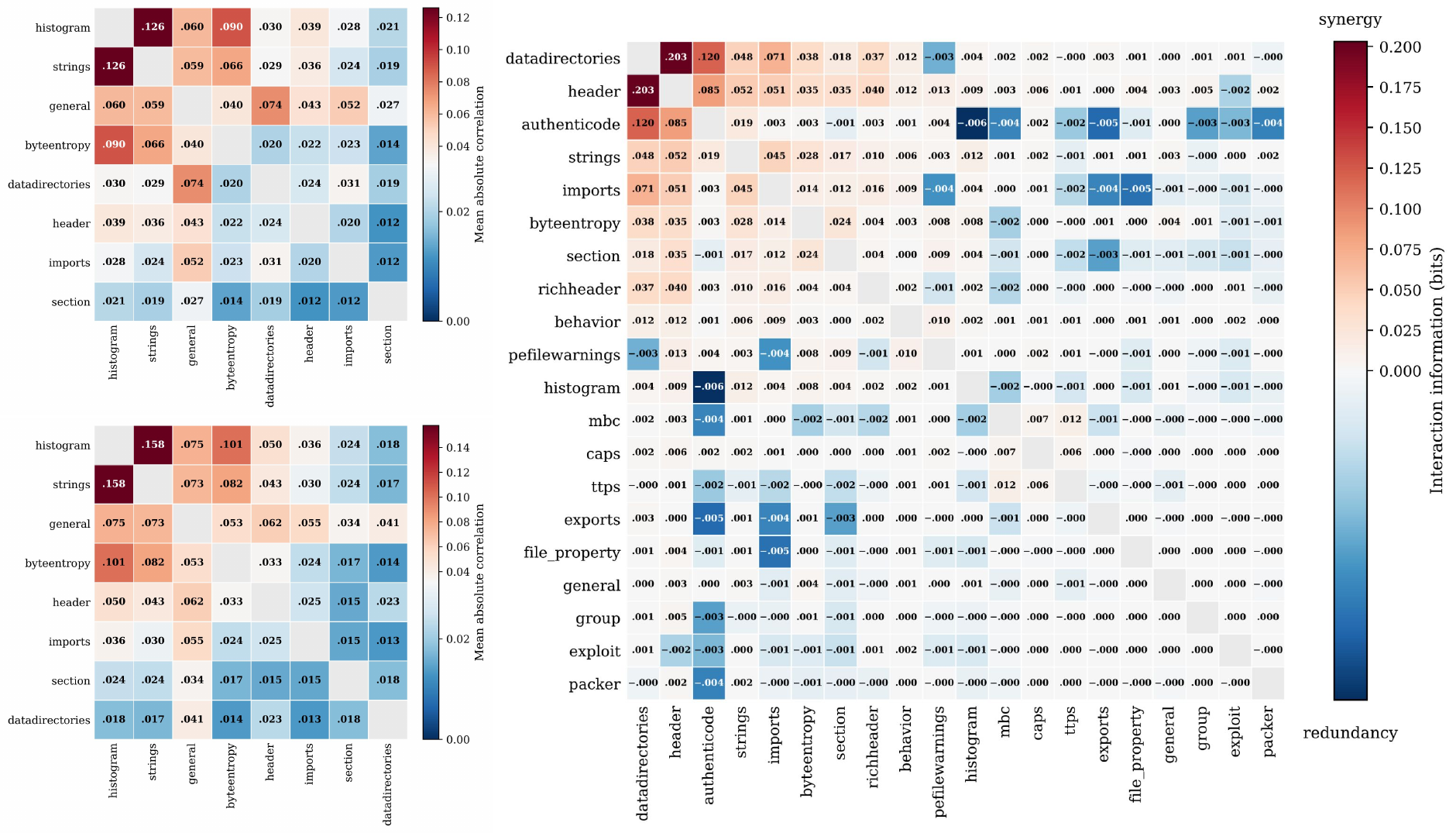}}

    \caption{Group-level dependence of the static PE feature space. \textbf{Left :} mean absolute correlation for EMBER-2018 (Top) and BODMAS (Bottom), both dominated by a redundant \texttt{histogram}, \texttt{strings}, and \texttt{byte-entropy} cluster. \textbf{Right:} pairwise interaction information (bits) for EMBER-2024, dominated by synergy among the structural groups \texttt{data-directories}, \texttt{header}, \texttt{authenticode}, and \texttt{imports}. \textbf{Takeaway:} the groups are not independent;  redundancy violates independence in EMBER-2018/BODMAS (hurting conditional SHAP, Lemmas~2--4), synergy in EMBER-2024 (hurting interventional SHAP, Lemma~1).}
    \label{fig:heatmap}
\end{figure*}

\textbf{Feature Redundancy:} In the proxy redundancy experiment, we added exact duplicates of features, $C_j=X_j$, and checked whether model performance remained stable as SHAP credit moved from the original feature to its copy. The proxy redundancy experiment verifies Lemma~\ref{lem:redundancy_failure_conditional_shap_malware}.

\subsection{Baseline}
Figure~\ref{fig:baseline} reports the performance of the two fixed detectors, LightGBM and XGBoost, across EMBER-2018, EMBER-2024, and BODMAS, as bootstrap means with $95\%$ confidence intervals over five seeds. All six detectors are uniformly strong. LightGBM attains an AUC of $0.990$ (precision $0.991$) on EMBER-2018, $\texttt{0.964}$ (precision $\texttt{0.972}$) on EMBER-2024, and $0.996$ (precision $0.995$) on BODMAS, while XGBoost reaches $0.990$ (precision $0.943$) on EMBER-2018, $\texttt{0.953}$ (precision $\texttt{0.962}$) on EMBER-2024, and $0.996$ (precision $0.994$) on BODMAS. Because the classifiers are already accurate, any subsequent change in SHAP attribution must be due to the explanation settings, the value function, the background distribution, proxy features, or invalid coalitions, rather than poor model quality.

\begin{figure}[t]
    \centerline{\includegraphics[width=0.9\columnwidth]{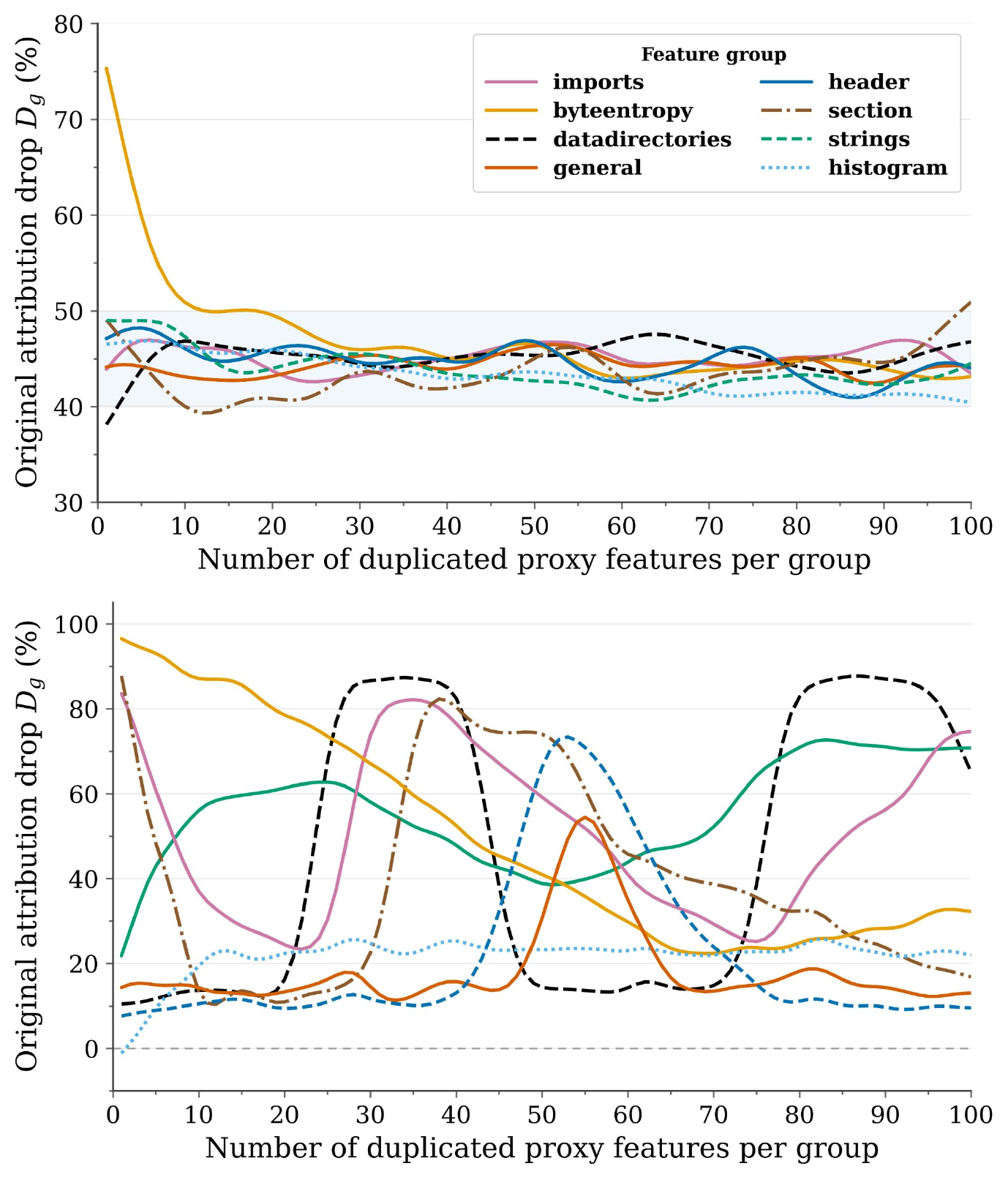}}    
    \caption{Proxy-redundancy sweep on EMBER-2018. \emph{X:} injected proxies per group ($0$--$100$); \emph{Y:} drop in original-feature attribution (\%). \textbf{Top (LightGBM):} smooth convergence to $40$--$50\%$, matching the $m=2$ prediction $1-1/m=0.50$ ($D_g\approx0.495$). \textbf{Bottom (XGBoost):} the same failure, but unstable and non-monotonic. Accuracy is unchanged; only SHAP allocation shifts. \textbf{Takeaway:} redundant proxies divert credit from model-used features without affecting predictions, demonstrating the redundancy failure of Lemma~2.}    
    \label{fig:proxy-plots}
\end{figure}

\subsection{Dependency}

Figure~\ref{fig:heatmap} reports the group-level dependence structure of the static PE feature space across three benchmarks. The two left panels measure the mean absolute correlation between every pair of feature groups for EMBER-2018 (upper Left) and BODMAS (lower Left). In contrast, the right panel reports pairwise interaction information (in bits) for EMBER-2024. In both correlation maps, the dominant structure is a content-based cluster binding the \texttt{histogram}, \texttt{strings}, and \texttt{byte-entropy} groups, with a weaker but consistent association extending to the \texttt{section} and \texttt{header} groups. BODMAS shows this clustering more strongly than EMBER-2018, most clearly in the histogram-strings relationship. These maps measure statistical association rather than causation, and the natural reading is that the clustered groups share common latent sources, such as packing, file size, compiler behavior, or family-specific implementation style, which simultaneously imprint the byte frequency, entropy, and string statistics extracted from the same executable. Operationally, for many feature subsets $S$, the data satisfy $P(X_{\bar{S}} \mid X_S = x_S) \neq P(X_{\bar{S}})$; the feature groups are not the independent players SHAP's coalition mechanism presupposes.

The EMBER-2024 panel exposes a complementary and equally damaging mode of dependence. Rather than the redundancy captured by correlation, the interaction information quantifies how much two groups jointly determine the prediction beyond their individual contributions. Positive values denote \emph{synergy}, where the groups carry information only when observed together, and negative values denote \emph{redundancy}, where they duplicate the same signal. EMBER-2024 is dominated by synergy rather than redundancy, with the largest interactions ($\approx 0.20$ bits) concentrated among the PE \emph{structural} groups \texttt{data-directories}, \texttt{header}, \texttt{authenticode}, and \texttt{imports} that jointly encode how the file is laid out and signed. This is the information-theoretic counterpart of the correlation clusters in the left panels, where the content groups of EMBER-2018 and BODMAS are \emph{redundantly} dependent; the structural groups of EMBER-2024 are \emph{synergistically} dependent. Both patterns violate the separable player assumption, but they fail in opposite directions, and a single correlation statistic would not have revealed the synergistic case at all.

These panels provide the empirical basis for our position and lemmas. The redundancy visible in the EMBER-2018 and BODMAS content clusters is exactly the regime of Lemma~\ref{lem:redundancy_failure_conditional_shap_malware}. When several groups encode the same latent factor $Z$, conditioning on any one of them reveals the others; thus, conditional SHAP dilutes a model used feature's credit by $1/m$ and transfers it to proxies that the classifier never uses, while the more general statistical association across these groups drives the proxy attribution and sign-instability failures of Lemmas~\ref{lem:proxy_attribution_failure_conditional_shap_malware} and~\ref{lem:sign_instability_conditional_shap_malware}. The synergy in EMBER-2024, in turn, indicts the interventional value function because synergistic groups are informative only in combination. Replacing one with an independent background draw destroys the joint configuration. It produces an off-manifold coalition that no real executable would exhibit, so the model is queried in regions it never learned. The figure therefore shows that neither realization of SHAP escapes the data structure: redundancy undermines the conditional formulation, and synergy undermines the interventional one. This is why the chosen coalition game, rather than malware behavior, governs attribution and why SHAP cannot serve as a standalone explanation of these detectors.

\subsection{Proxy Redundancy}

Figure ~\ref{fig:proxy-plots}, top plot, shows the proxy redundancy for EMBER-2018 with the LightGBM model. The LightGBM panel offers the closest empirical match to Lemma~\ref {lem:redundancy_failure_conditional_shap_malware}. In the pairwise case, which instantiates the $m=2$ regime of the lemma, each feature used by the model is paired with a single exact proxy, so the original features are predicted to retain a fraction $1/m = 1/2$ of their attribution, resulting in an Original Attribution Loss of $1 - 1/m = 0.50$. Defining $\mathrm{D_g} = 1 - \big(\sum_i A_i^{\mathrm{aug}}\big)\big/\big(\sum_i A_i^{\mathrm{base}}\big)$ over the $20$ original/redundant pairs, we measure $\mathrm{D_g} = 1 - 2.061/4.08 \approx 0.495$, within half a percentage point of the closed form value. At the same time, the displaced credit reappears on the proxy copies ($\sum_i A_{c_i}^{\mathrm{aug}} = 2.262$) that the classifier does not functionally require. As the sweep increases the number of injected duplicates from $0$ to $100$, the curves converge to a stable band in the $40$--$50\%$ range, with an overall mean original drop of about $44.6\%$ and a median of about $44.3\%$ across the first $100$ settings. This smooth convergence is consistent with the lemma: the idealized $1/m$ collapse holds for the exact conditional value function, whereas LightGBM, when applied to a fitted gradient-boosted ensemble, routes splits through a bounded set of redundancy representatives. Thus, additional exact copies enter few or no splits, and the estimator effectively continues to observe the two-player regime ($m_{\mathrm{eff}}\!\approx\!2$). The detector is held fixed, and its predictive performance is unchanged; only the SHAP allocation moves. LightGBM demonstrates the redundancy failure of Lemma~\ref{lem:redundancy_failure_conditional_shap_malware}: roughly half of the original attribution is reassigned to redundant proxies, although the model and underlying malware signal are unchanged.

The XGBoost panel (Figure~\ref{fig:proxy-plots}, bottom plot) exhibits the same qualitative failure predicted by Lemma~\ref{lem:redundancy_failure_conditional_shap_malware}, a substantial transfer of attribution from original features to redundant proxies, but in a markedly less stable form. Rather than the smooth $40$--$50\%$ convergence seen for LightGBM, the XGBoost curves are highly non-monotonic across the sweep. Some feature groups drop by more than $80$--$90\%$, while others remain low or briefly register negative drops, in which the original features receive \emph{more} attribution after proxy augmentation than before. The overall mean original drop across the first $100$ settings is about $37.9\%$, but the median is only about $28.9\%$, reflecting a far more uneven distribution than LightGBM. This behavior remains consistent with Lemma~\ref{lem:redundancy_failure_conditional_shap_malware} rather than contradicting it. The lemma fixes the redundancy mechanism but not the magnitude of the realized split, which is determined by how the estimator distributes credit among observationally equivalent features; the negative and oscillating drops do not indicate that any proxy became causally important, but rather that XGBoost's greedy split selection and tie-breaking among exact duplicates redistribute credit in an unstable, path-dependent manner. The contrast with LightGBM is itself the point: for the same datasets, the same duplicate construction, and equally strong fixed detectors, the amount of attribution that moves and its stability depend on the tree learner and the SHAP allocation rather than on the malware. XGBoost thus reinforces the paper's central position that SHAP credit is not invariant to redundant feature encoding, while showing that the instability the lemma anticipates can manifest far more severely than the idealized $1/m$ value suggests.

\section{Conclusion}
SHAP should not be treated as a standalone explanation of static PE malware classifiers. For a fixed value function and simplified input mapping, its attributions are formally valid and satisfy local accuracy, missingness, and consistency; however, these guarantees describe a chosen coalition game, not the malicious behavior an analyst seeks to understand. Because PE feature groups are not independent players but linked parts of the same executable, jointly shaped by file structure, packing, compiler behavior, and family conventions, the choice of game is neither neutral nor obvious and materially determines the result. Conditional SHAP dilutes a model-used feature's credit by a factor of $1/m$ across $m-1$ redundant proxies, assigns nonzero importance to features the classifier never uses, and can reverse an unused feature's sign as the data distribution changes, while interventional SHAP queries off-manifold coalitions that no real executable would exhibit. Our experiments confirm this account while isolating it from model quality. We therefore position SHAP as a constrained diagnostic requiring an explicitly stated data distribution and domain validation, rather than a self-sufficient account of malware behavior \cite{gaur2022knowledge}.

\bibliographystyle{IEEEtran}
\bibliography{biblio}

\end{document}